\documentclass{article}
\usepackage[affil-it]{authblk}

\usepackage{graphicx}
\usepackage{amsmath}
\usepackage{amssymb}
\usepackage{amsfonts}
\usepackage{amsthm}
\usepackage{natbib}
\usepackage{setspace}
\usepackage{verbatim}
\usepackage{times}
\usepackage{helvet}
\usepackage{courier}
\usepackage{bm}
\usepackage{url}
\usepackage{dcolumn}
\usepackage{multirow}
\usepackage{xcolor}
\usepackage{color}
\usepackage{epsfig}

\newcommand{\no}[1]{}

\newtheorem{theorem}{Theorem}
\newtheorem{cor}{Corollary}
\newtheorem{lem}{Lemma}
\newtheorem{fact}{Fact}

\newcommand{\leftchildarg}[1]{\text{left-child}(\ensuremath{#1})}
\newcommand{\rightchildarg}[1]{\text{right-child}(\ensuremath{#1})}
\newcommand{\parentarg}[1]{\text{parent}(\ensuremath{#1})}
\newcommand{\subtreesizearg}[1]{\text{subtree-size}(\ensuremath{#1})}
\newcommand{\lcaarg}[2]{\text{LCA}(\ensuremath{#1,#2})}
\newcommand{\levelancarg}[2]{\text{level-ancestor}(\ensuremath{#1,#2})}
\newcommand{\distancearg}[2]{\text{distance}(\ensuremath{#1,#2})}
\newcommand{\degreearg}[1]{\ensuremath{\text{degree}(#1)}}
\newcommand{\deptharg}[1]{\ensuremath{\text{depth}(#1)}}
\newcommand{\heightarg}[1]{\ensuremath{\text{height}(#1)}}
\newcommand{\lmleafarg}[1]{\ensuremath{\text{left-most-leaf}(#1)}}
\newcommand{\rmleafarg}[1]{\ensuremath{\text{right-most-leaf}(#1)}}
\newcommand{\leafselectarg}[1]{\ensuremath{\text{leaf-select}(#1)}}
\newcommand{\leafrankarg}[1]{\ensuremath{\text{leaf-rank}(#1)}}

\newcommand{\rankarg}[3]{\ensuremath{\text{rank}_{#1}(#2,#3)}}
\newcommand{\selectarg}[3]{\ensuremath{\text{select}_{#1}(#2,#3)}}
\newcommand{\leftdepth}[1]{\ensuremath{\text{Ldepth}(#1)}}
\newcommand{\rightdepth}[1]{\ensuremath{\text{Rdepth}(#1)}}
\newcommand{\leftleaves}[1]{\ensuremath{\text{Lleaves}(#1)}}

\newcommand{\levelrmostarg}[1]{\ensuremath{\text{level-left-most}(#1)}}
\newcommand{\levellmostarg}[1]{\ensuremath{\text{level-right-most}(#1)}}
\newcommand{\levelsucarg}[1]{\ensuremath{\text{level-successor}(#1)}}
\newcommand{\levelpredarg}[1]{\ensuremath{\text{level-predecessor}(#1)}}
\newcommand{\childrankarg}[1]{\ensuremath{\text{child-rank}(#1)}}
\newcommand{\noderankarg}[2]{\text{node-rank}\ensuremath{_{#1}(#2)}}
\newcommand{\nodeselectarg}[2]{\text{node-select}\ensuremath{_{#1}(#2)}}

\newcommand{\rank}{rank}
\newcommand{\select}{select}

\newcommand{\noderank}[1]{\text{node-rank}\ensuremath{_{#1}}}
\newcommand{\nodeselect}[1]{\text{node-select}\ensuremath{_{#1}}}

\newcommand{\preorder}{\text{preorder}}
\newcommand{\postorder}{\text{postorder}}
\newcommand{\inorder}{\text{inorder}}
\newcommand{\dfuds}{\text{DFUDS}}

\newcommand{\range}[2]{\ensuremath{[#1..#2]}}
\mathchardef\mhyphen="2D
\newcommand{\secondmin}{\ensuremath{\text{R2M}}}
\newcommand{\secondminarg}[2]{\ensuremath{\text{R2M}(#1,#2)}}
\newcommand{\toptwo}{\ensuremath{\text{RT2Q}}}
\newcommand{\toptwoarg}[2]{\ensuremath{\text{RT2Q}(#1,#2)}}
\newcommand{\rmqarg}[2]{\ensuremath{\text{RMQ}(#1,#2)}}
\newcommand{\argsecondmin}[3]{\ensuremath{\text{argmin}\{A[#1] : {#1\in (\range{#2}{#3}\setminus \rmqarg{#2}{#3})}\}}}
\newcommand{\lispine}[1]{\ensuremath{\text{lispine}(#1)}}
\newcommand{\rispine}[1]{\ensuremath{\text{rispine}(#1)}}
\newcommand{\lspine}[1]{\ensuremath{\text{lspine}(#1)}}
\newcommand{\rspine}[1]{\ensuremath{\text{rspine}(#1)}}

\newcommand{\cartesian}[1]{\ensuremath{T_{#1}}}

\begin{document}


\title{Encoding Range Minimum Queries\thanks{An extended abstract of some of the results in Sections 1 and 2 appeared in \emph{Proc. 18th Annual International Conference on Computing and Combinatorics (COCOON 2012)}, Springer LNCS 7434, pp. 396--407.}}

\author{
Pooya Davoodi\thanks{Research supported by NSF grant CCF-1018370 and BSF grant 2010437.} \\
Polytechnic Institute of New York University, United States \\
\texttt{pooyadavoodi@gmail.com}
\and
Gonzalo Navarro\thanks{Partially funded by Millennium Nucleus Information and Coordination in Networks ICM/FIC P10-024F, Chile.}\\
Department of Computer Science, University of Chile, Chile \\
\texttt{gnavarro@dcc.uchile.cl}
\and
Rajeev Raman\\
Department of Computer Science, University of Leicester, UK\\
\texttt{rr29@leicester.ac.uk}
\and
S. Srinivasa Rao\thanks{Research partly supported by Basic Science Research Program through the National Research Foundation of Korea funded by the Ministry of Education, Science and Technology (Grant number 2012-0008241).}\\
School of Computer Science and Engineering, Seoul National University, Republic of Korea\\
\texttt{ssrao@cse.snu.ac.kr}
}

\date{}
\maketitle

\begin{abstract}
We consider the problem of \emph{encoding} \emph{range minimum queries}
(RMQs): given an array $A\range{1}{n}$ of distinct totally ordered values, to
pre-process $A$ and create a data structure that can answer the query \rmqarg{i}{j}, which
returns the index containing the smallest element in $A\range{i}{j}$,
\emph{without} access to the array $A$ at query time.  We give
a data structure whose space usage is $2n + o(n)$ bits, which is asymptotically
optimal for worst-case data, and answers RMQs in $O(1)$ worst-case time.  This
matches the previous result of Fischer and Heun, but is obtained in a more natural
way.  Furthermore,
our result can encode the RMQs of a random array $A$ in $1.919n + o(n)$ bits in
expectation,
which is not known to hold for Fischer and Heun's result.  We then
generalize our result
to the encoding \emph{range top-2 query} (RT2Q) problem, which is like the encoding RMQ
problem except that
the query \toptwoarg{i}{j} returns the indices of both the smallest and
second-smallest
elements of $A\range{i}{j}$.
We introduce a data structure using $3.272n+o(n)$ bits that answers RT2Qs
in constant time, and also give lower bounds on the 
\emph{effective entropy} of~RT2Q.
\end{abstract}



\section{Introduction}
\label{sec:intro}
Given an array $A\range{1}{n}$ of elements from a totally ordered set, 
the \emph{range minimum query} (RMQ) problem is to pre-process $A$ 
and create a data structure so that the query \rmqarg{i}{j}, which
takes two indices $1 \le i \le j \le n$
and returns $\mbox{\rm argmin}_{i \le k \le j} A[k]$, is supported efficiently (both in terms of space and time).  
We consider the \emph{encoding} version
of this problem: after pre-processing $A$, the data structure should answer RMQs \emph{without} access to $A$; in other words, the data structure should encode all the information about $A$ needed to answer RMQs. In many applications that deal with storing and indexing massive data, 
the values in $A$ have no intrinsic significance and $A$ can be discarded
after pre-processing (for example, $A$ may contain
scores that are used to determine the relative order of 
documents returned in response to a search query).  As we now
discuss, the encoding of $A$ for RMQs can often take much less space than $A$ itself, so 
encoding RMQs can facilitate the  
efficient in-memory processing of massive data.

It is well known \cite{gbt-stoc84} that
the RMQ problem is equivalent to the problem of supporting lowest common ancestor (LCA) queries on a binary tree, the \emph{Cartesian} tree of $A$ \cite{Vuillemin1980}.
The Cartesian tree of $A$ is a binary tree with $n$ nodes, in which the root is labeled by $i$ where $A[i]$ is the minimum element in $A$; the left subtree of of the root is the Cartesian tree of $A\range{1}{i-1}$ and the right subtree of the root is the Cartesian tree of $A\range{i+1}{n}$. The answer to \rmqarg{i}{j} is the label of the LCA 
of the nodes labeled by~$i$ and~$j$.  Thus, knowing the topology of the Cartesian tree of $A$ suffices to answer RMQs on $A$.

Farzan and Munro \cite{fm-algo12} showed that an $n$-node binary tree can be represented in
$2n + o(n)$ bits, while supporting LCA queries in $O(1)$ time\footnote{The
time complexity of this result assumes the word RAM model with logarithmic
word size, as do all subsequent results in this paper.}. Unfortunately, this does not solve the RMQ problem. The difficulty 
is that nodes in the Cartesian tree are labelled with the index of the corresponding 
array element, which is equal to the node's rank in the inorder traversal of the 
Cartesian tree. A common feature
of \emph{succinct} tree representations, 
		such as that of \cite{fm-algo12},
is that they do not allow the user to specify
the numbering of nodes \cite{ianfest-survey}, and while 
existing succinct binary tree representations support numberings such as level-order \cite{jacobson89}
and preorder \cite{fm-algo12}, they do not support inorder. Indeed, for this reason, 
Fischer and Heun  \cite{fh-sjc11} solved
the problem of optimally encoding RMQ via an ordered rooted tree, rather than
via the more natural Cartesian tree.  

Our first contribution is to describe how, using $o(n)$ additional bits, we can 
add the functionality below to the $2n + o(n)$-bit representation of Farzan and Munro:
\begin{itemize}
\item \noderankarg{\inorder}{x}: returns the position in inorder of node $x$.
\item \nodeselectarg{\inorder}{y}: returns the node $z$ whose inorder position is $y$.
\end{itemize}
Here, $x$ and $z$ are node numbers in the node numbering scheme of Farzan and Munro,
and both operations take $O(1)$ time.  Using this, we can encode RMQs of an 
array $A$ using $2n + o(n)$ bits, and answer RMQs in $O(1)$ time as follows.  
We represent the Cartesian tree of $A$ using the representation
of Farzan and Munro, augmented with the above operations, and answer \rmqarg{i}{j} as
$$\rmqarg{i}{j} =
\noderankarg{\inorder}{\lcaarg{\nodeselectarg{\inorder}{i}}{\nodeselectarg{\inorder}{j}}}.$$
We thus match asymptotically the 
result of Fischer and Heun \cite{fh-sjc11}, while using a more
direct approach.  Furthermore,
using our approach, we can encode RMQs of a random permutation using $1.919n + o(n)$ bits in
expectation and answer RMQs in $O(1)$ time. It is not clear how to obtain this
result using the approach of Fischer and Heun.


Our next contribution is to consider a generalization of RMQs, namely, to pre-process
a totally ordered array $A\range{1}{n}$ to answer \emph{range top-2 queries} (\toptwo{}). 
The query \toptwoarg{i}{j} returns the indices of the minimum as well as the second minimum values in $A\range{i}{j}$.  Again, we consider the encoding
version of the problem, so that the data structure does not have access to $A$ when answering a query.  Encoding problems, such as the 
RMQ and \toptwo{}, are fundamentally
about determining the \emph{effective entropy}
of the data structuring problem \cite{GIKRRS12}.  Given
the input data drawn from a set of inputs ${\cal S}$, 
and a set of queries $Q$, the effective
entropy of $Q$ is $\lceil \log_2 |{\cal C}| \rceil$, where
${\cal C}$ is the set of equivalence
classes on ${\cal S}$ induced by $Q$, whereby two objects 
from ${\cal S}$ are equivalent if they provide the same answer to all queries in~$Q$.  For the RMQ problem, it is easy to see that every binary tree is the Cartesian tree of some array $A$.
Since there are $C_n = \frac{1}{n+1}{{2n} \choose {n}}$ 
$n$-node binary trees, the effective entropy
of RMQ is exactly $\lceil \log_2 C_n \rceil = 2n - O(\log n)$ bits.

The effective entropy of the more general 
\emph{range top-$k$ problem}, or finding the indices of the $k$ smallest
elements in a given range $A[i,j]$, was recently 
shown to be $\Omega(n \log k)$ bits by Grossi et al. 
\cite{GINRS13}.  However, for $k = 2$, their approach only
shows that the effective entropy of \toptwo{} is $\ge n/2$ --
much less than the effective entropy of RMQ.  Using an
encoding based upon  merging paths in Cartesian trees, 
we show that the effective entropy of \toptwo{} is at least
$2.638n-O(\log n)$ bits. We show that this effective entropy applies also
to the (apparently) easier problem of returning just the second minimum
in an array interval, \secondmin{}$(i,j)$.
We complement this result by giving a data structure for encoding 
\toptwo{s} that takes $3.272n + o(n)$ bits and answers queries in
$O(1)$ time.  This structure
builds upon our new $2n + o(n)$-bit RMQ encoding by adding further
functionality to the binary tree representation of Farzan and Munro.
We note that the   
range top-$k$ encoding of Grossi et al.\ \cite{GINRS13} builds upon a encoding that
answers \toptwo{} in $O(1)$ time, but their encoding for this 
subproblem uses $6n+o(n)$ bits.

\subsection{Preliminaries}

Given a bit vector $B\range{1}{m}$, \rankarg{B}{1}{i} returns the number of 1s in $B\range{1}{i}$, and \selectarg{B}{1}{i} returns the position of the $i$th 1 in $B$. The operations \rankarg{B}{0}{i} and \selectarg{B}{0}{i} are defined analogously for 0s. A data structure that supports the operations \rank\ and \select\ is a building block of many succinct data structures. The following lemma states a \rank-\select\ data structure that we use to obtain our results.

\begin{lem}\cite{Cla96,Mun96}
\label{lem:rank-select}
Given a bit vector $B\range{1}{m}$, there exists a data structure of size $m+o(m)$ bits that supports \rankarg{B}{1}{i}, \rankarg{B}{0}{i} \selectarg{B}{1}{i}, and \selectarg{B}{0}{i} in $O(1)$ time.
\end{lem}

We also utilize the following lemma, which states a more space-efficient \rank-\select\ data structure that assumes the number of 1s in $B$ is known.

\begin{lem}\cite{rrr-talg07}
\label{lem:fully-index-dic}
Given a bit vector $B\range{1}{m}$ that contains $n$ 1s, there exists a data structure of size $\log {m \choose n} + o(m)$ bits, that supports \rankarg{B}{1}{i}, \rankarg{B}{0}{i} \selectarg{B}{1}{i}, and \selectarg{B}{0}{i} in $O(1)$ time.
\end{lem}

%
%

\section{Representation Based on Tree Decomposition}
\label{sec:bin-rep-inorder}

We now describe a succinct representation of binary trees that supports a
comprehensive list of operations \cite{hms-icalp07,fm-swat08,fm-algo12}.%
\footnote{\label{footnote:operations} This list includes \leftchildarg{v},
\rightchildarg{v}, \parentarg{v}, \childrankarg{v}, 
\degreearg{v}, \subtreesizearg{v}, \deptharg{v},
\heightarg{v}, \lmleafarg{v}, \rmleafarg{v}, \leafrankarg{v},
\leafselectarg{j}, \levelancarg{v}{i}, \lcaarg{u}{v}, \distancearg{u}{v},
\levellmostarg{i}, \levelrmostarg{i}, \levelsucarg{v}, and \levelpredarg{v},
where $v$ denotes a node, $i$ denotes a level, and $j$ is an integer. Refer to
the original articles~\cite{hms-icalp07,fm-algo12} for the definition of these operations.}. 
The structure of Farzan and Munro~\cite{fm-algo12} supports multiple orderings on the nodes of the tree including \preorder, \postorder, and \dfuds\ order by providing the operations \noderankarg{\preorder}{v}, \nodeselectarg{\preorder}{v}, \noderankarg{\postorder}{v}, \nodeselectarg{\postorder}{v}, \noderankarg{\dfuds}{v}, and \nodeselectarg{\dfuds}{v}. We provide two additional operations \noderankarg{\inorder}{v} and \nodeselectarg{\inorder}{v} thereby
also supporting \inorder\ numbering on the nodes.


Our data structure consists of two parts: $(a)$ the data structure of Farzan
and Munro~\cite{fm-algo12}, and $(b)$ an additional structure we construct to
specifically support \noderank{\inorder} and \nodeselect{\inorder}. In the
following, we outline the first part (refer to Farzan and Munro
\cite{fm-algo12} for more details), and then we explain in detail the second part.

\subsection{Succinct cardinal trees of Farzan and Munro \cite{fm-algo12}}
\label{sec:ds-fm-algo12}
Farzan and Munro~\cite{fm-algo12} reported a succinct representation of
cardinal trees ($k$-ary trees). Since binary trees are a special case of
cardinal trees (when $k=2$), their data structure can be used as a succinct
representation of binary trees. The following lemma states their result for
binary trees:%

\begin{lem}\cite{fm-algo12}
\label{lem:ds-fm-algo12}
A binary tree with $n$ nodes can be represented using $2n+o(n)$ bits of space,
while a comprehensive list of operations \cite[Table 2]{fm-algo12} (or see
Footnote~\ref{footnote:operations}) can be supported in $O(1)$ time.
\end{lem}

This data structure is based on a tree decomposition similar to previous ones~\cite{grr-atalg06,hms-icalp07,mrs-soda01}. An input binary tree is first partitioned into $O(n/\log^2 n)$ \emph{mini-trees} each of size at most~$\lceil\log^2 n\rceil$, that are disjoint aside from their roots.
Each mini-tree is further partitioned (recursively) into $O(\log n)$ \emph{micro-trees} of size at most $\lceil\frac{\lg n}{8}\rceil$, which are also disjoint aside from their roots. A non-root node in a mini-tree $t$, that has a child located in a different mini-tree, is called a \emph{boundary node} of $t$ (similarly for micro-trees). 

The decomposition algorithm achieves the following prominent property: each mini-tree has at most one boundary node and each boundary node has at most one child located in a different mini-tree (similar property holds for micro-trees).
%
%
This property implies that aside from the edges on the mini-tree roots, there is at most one edge in each mini-tree that connects a node of the mini-tree to its child in another mini-tree. These properties also hold for micro-trees.

It is well-known that the topology of a tree with $k$ nodes can be described with a fingerprint of size~$2k$ bits. Since the micro-trees are small enough, the operations within the micro-trees can be performed by using a universal lookup-table of size $o(n)$ bits, where the fingerprints of micro-trees are used as indexes into the table. 

The binary tree representation consists of the following parts (apart from the lookup-table): 1) representation of each micro-tree: its size and fingerprint; 2) representation of each mini-tree: links between the micro-trees within the mini-tree; 3) links between the mini-trees. The overall space of this data structure is $2n+o(n)$ bits \cite{fm-algo12}.




\subsection{Data structure for \noderank{\inorder} and \nodeselect{\inorder}}

We present a data structure that is added to the structure of Lemma \ref{lem:ds-fm-algo12} in order to support \noderank{\inorder} and \nodeselect{\inorder}. This additional data structure contains two separate parts, each to support one of the operations. In the following, we describe each of these two parts. Notice that we have access to the succinct binary tree representation of Lemma \ref{lem:ds-fm-algo12}.

\subsubsection{Operation \noderank{\inorder}}
We present a data structure that can compute the inorder number of a node $v$,
given its preorder number. To compute the inorder number of $v$, we compute
two values~$c_1(v)$ and~$c_2(v)$ defined as follows. Let $c_1(v)$ be the
number of nodes that are visited before $v$ in inorder traversal and visited
after $v$ in preorder traversal; and let $c_2(v)$ be the number of nodes that
are visited after $v$ in inorder traversal and visited before $v$ in preorder
traversal (our method below to compute $c_2(v)$ is also utilized in Section
\ref{sec:secondmin} to perform an operation called \leftdepth{v}, which computes $c_2(v)$ for any given node $v$). Observe that the inorder number of $v$ is equal to its preorder number plus~$c_1(v) - c_2(v)$. 

The nodes counted in $c_1(v)$ are all the nodes located in the left subtree of $v$, which can be counted by subtree size of the left child of $v$. The nodes counted in $c_2(v)$ are all the ancestors of $v$ whose left child is also on the $v$-to-root path, i.e., $c_2(v)$ is the number of left-turns in the $v$-to-root path.
We compute $c_2(v)$ in a way similar to computing the depth of a node as follows. For the root $r_m$ of each mini-tree, we precompute and store~$c_2(r_m)$ which requires $O((n/\log^2 n) \log n)=o(n)$ bits. 
Let mini-$c_2(v)$ and micro-$c_2(v)$ be the number of left turns from a node $v$ up to only the root of respectively the mini-tree and micro-tree containing $v$. For the root $r_\mu$ of each micro-tree, we precompute and store mini-$c_2(r_\mu)$. We use a lookup table to compute micro-$c_2(v)$ for every node~$v$. 

Finally, to compute $c_2(v)$, we simply calculate $c_2(r_m) + \text{mini-}c_2(r_\mu) + \text{micro-}c_2(v)$, where $r_m$ and $r_\mu$ are the root of respectively the mini-tree and micro-tree containing $v$. The data structure of Lemma \ref{lem:ds-fm-algo12} can be used to find $r_m$ and $r_\mu$ and the calculation can be done in $O(1)$ time.

\subsubsection{Operation \nodeselect{\inorder}}
We present a data structure that can compute the preorder number of a node $v$, given its inorder number. To compute the preorder number of $v$, we compute 1) the preorder number of the root $r_m$ of the mini-tree containing $v$; and 2) $c(v,r_m)$: the number of nodes that are visited after $r_m$ and before $v$ in preorder traversal, which may include nodes both within and outside the mini-tree rooted at $r_m$. Observe that the preorder number of $v$ is equal to the preorder number of $r_m$ plus $c(v,r_m)$. In the following, we explain how to compute these two quantities:

(1) We precompute the preorder numbers of all the mini-tree roots and store
them in~$P\range{0}{n_m-1}$ in some arbitrary order defined for mini-trees,
where $n_m=O(n/\log^2 n)$ is the number of mini-trees. Notice that each
mini-tree now has a \emph{rank} from \range{0}{n_m-1}. Later on, when we want
to retrieve the preorder number of the root of the mini-tree containing $v$,
we only need to determine the rank $i$ of the mini-tree and read the answer
from $P[i]$. In the following, we explain a data structure that supports finding the rank of the mini-tree containing any given node $v$.

In the preprocessing, we construct a bit-vector $A$ and an array $B$ of mini-tree ranks, which are initially empty, by traversing the input binary tree in inorder as follows (see Figure \ref{fig:node-select} for an example):

For $A$, we append a bit for each visited node and thus the length of $A$ is $n$. If the current visited node and the previous visited node are in two different mini-trees, then the appended bit is $1$, and otherwise $0$; if a mini-tree root is common among two mini-trees, then its corresponding bit is $0$ (i.e., the root is considered to belong to the mini-tree containing its left subtree since a common root is always visited after its left subtree is visited); the first bit of $A$ is always a $1$. 

For $B$, we append the rank of each visited mini-tree; more precisely, if the current visited node and the previous visited node are in two different mini-trees, then we append the rank of the mini-tree containing the current visited node, and otherwise we append nothing. Similarly, a common root is considered to belong to the mini-tree containing its left subtree;  the first rank in $B$ is the rank of the mini-tree containing the first visited node. 

\begin{figure}[t]
  \centering
  \includegraphics[width=110mm]{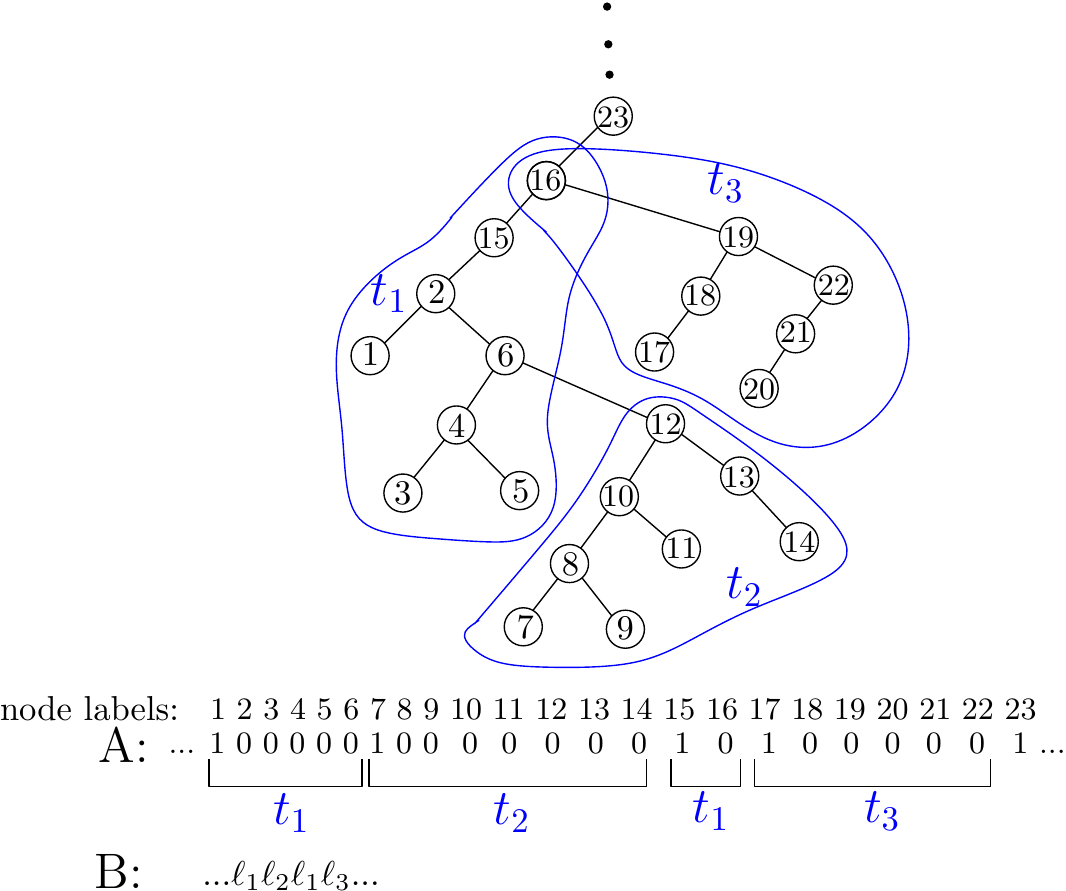}
  \caption{Figure depicts a part of a binary tree where $t_1$, $t_2$, and $t_3$ are its three mini-trees. Node labels are in inorder ordering of these nodes. Each node has a corresponding bit in $A$ and each mini-tree has one or two corresponding labels ($\ell_i$ is the label of $t_i$) in $B$.}
  \label{fig:node-select}
\end{figure}

We observe that a node $v$ with inorder number $i$ belongs to the mini-tree with rank $B[\rankarg{A}{1}{i+1}]$, and thus~$P[B[\rankarg{A}{1}{i+1}]]$ contains the preorder number of the root of the mini-tree containing $v$.

We represent $A$ using the data structure of Lemma~\ref{lem:fully-index-dic},
which supports \rank\ in constant time. In order to analyze the space, we
prove that the number of $1$s in $A$ is at most~$2 n_m$: each mini-tree has at
most one edge leaving the mini-tree aside from its root, which means that the
traversal can enter or re-enter a mini-tree at most twice. Therefore, the
space usage is $\lg{n \choose 2n_m} + o(n)= o(n)$ bits, as $n_m=O(n/\log^2 n)$.
We store $P$ and $B$ explicitly with no preprocessing on them. The length of $B$ is also at most~$2 n_m$ by the same argument. Thus, both $P$ and $B$ take $O(n/\log^2 n \cdot \log n)=o(n)$ bits.

(2) Let $S$ be the set of nodes that are visited after $r_m$ and before $v$ in the preorder traversal of the tree. Notice that $c(v,r_m)=|S|$. Let $t_m$ and $t_\mu$ be respectively the mini-tree and micro-tree containing $v$. We note that $S=S_1 \cup S_2 \cup S_3$, where $S_1$ contains the nodes of $S$ that are not in $t_m$, $S_2$ contains the nodes of $S$ that are in $t_\mu$, and $S_3$ contains the nodes that are in $t_m$ and not in $t_\mu$. Observe that $S_1$, $S_2$, and $S_3$ are mutually disjoint. Therefore, $c(v,r_m)=|S_1| + |S_2| + |S_3|$.
We now describe how to compute each size.

$S_1$: If $t_m$ has a boundary node which is visited before the root of $t_{\mu}$, then $|S_1|$ is the subtree size of the child of the boundary node that is out of $t_m$; otherwise $|S_1|=0$.

$S_2$: Since these nodes are within a micro-tree, $|S_2|$ can be computed using a lookup-table.

$S_3$: The local preorder number of the root of $t_\mu$, which results from traversing $t_m$ while ignoring the edges leaving $t_m$, is equal to $|S_3|$. We precompute the local preorder numbers of all the micro-tree roots. The method to store these local preorder numbers and the data structure that we construct in order to efficiently retrieve these numbers is similar to the part (1), whereas here a mini-tree plays the role of the input tree and micro-trees play the role of the mini-trees. In other words, we construct $P$, $A$, and $B$ of part (1) for each mini-tree. The space usage of this data structure is $o(n)$ bits by the same argument, regarding the fact that each local preorder number takes $O(\log\log n)$ bits.

\begin{theorem}
\label{thm:bin-rep-inorder}
A binary tree with $n$ nodes can be represented with a succinct data structure
of size $2n+o(n)$ bits, which supports \noderank{\inorder},
\nodeselect{\inorder}, plus a comprehensive set of operations \cite[Table 2]{fm-algo12}, all in $O(1)$ time.
\end{theorem}

\subsection{RMQs on Random Inputs}
The following theorem gives a slight generalization of Theorem~\ref{thm:bin-rep-inorder}, which uses entropy coding to exploit any differences in frequency between different types of nodes (Theorem~\ref{thm:bin-rep-inorder} corresponds to choosing all the $\alpha_i$s to be $1/4$ in the following):
\begin{theorem}
\label{thm:general}
For any positive constants $\alpha_0, \alpha_L, \alpha_R$ and $\alpha_2$, such that
$\alpha_0 + \alpha_L + \alpha_R + \alpha_2 = 1$, a binary tree with
$n_0$ leaves, $n_L$ ($n_R$) nodes with only a left (right) child and $n_2$ nodes with
both children can be represented using 
$\left ( \sum_{i \in \{0,L,R,2\}} n_i \lg (1/\alpha_i) \right) + o(n)$ bits of space, while a full set of operations \cite[Table 2]{fm-algo12} including 
\noderank{\inorder}, \nodeselect{\inorder} and \emph{LCA} can be supported in $O(1)$ time.
\end{theorem}
\begin{proof}
We proceed as in the proof of Theorem~\ref{thm:bin-rep-inorder}, but if
$\alpha = \min_{i \in \{0,L,R,2\}} \alpha_i$, we choose the size of the
micro-trees to be at most $\mu = \frac{\lg n}{2 \lg (1/\alpha)} = \Theta(\log
n)$. The $2n$-bit term in the representation of \cite{fm-algo12} comes from the
representation of the microtrees. Given a micro-tree with $\mu_i$ nodes of type $i$, for $i \in
\{0,L,R,2\}$ we encode it by writing the node types in level order (cf.
\cite{jacobson89}) and encoding this string using arithmetic coding with the
probability of a node of type $i$ taken to be $\alpha_i$.  The size of this
encoding is at most ~$\left(\sum_{i \in \{0,L,R,2\}} \mu_i \lg (1/\alpha_i)
\right)+2$ bits, from which the theorem follows. Note that our choice of $\mu$
guarantees that each microtree fits in $\frac{\lg n}{2}$ bits and thus can 
still be manipulated using universal look-up tables.
\end{proof}
\begin{cor}
If $A$ is a random permutation over $\{1,\ldots, n\}$, then RMQ queries on $A$ can be answered using
$(\frac{1}{3}+\lg 3)n + o(n) < 1.919n + o(n)$ bits in expectation.
\end{cor}
\begin{proof}
Choose $\alpha_0 = \alpha_2 = 1/3$ and $\alpha_R = \alpha_L = 1/6$. The claim 
follows from the fact that $\alpha_i n$ is the average value of $n_i$ on
random binary trees, for any $i \in \{0,L,R,2\}$ 
\cite[Theorem 1]{gikrs-isaac11}.
\end{proof}

While both our representation and that of Fischer and Heun \cite{fh-sjc11} 
solve RMQs in $O(1)$ time and use $2n+o(n)$ bits in the worst case, ours
allows an improvement in the average case. However, we are unable
to match the expected effective entropy of RMQs on random arrays $A$, which is
$\approx 1.736n+O(\log n)$ bits \cite[Thm.~1]{GIKRRS12} (see also 
\cite{Kieffer2009}).


It is natural to ask whether one can obtain improvements for the average case
via Fischer and Heun's approach \cite{fh-sjc11} as well.  Their approach first converts the
Cartesian tree to an \emph{ordinal tree} (an ordered, rooted tree) using the textbook transformation \cite{CLR}. 
%
To the best of our knowledge, the only ordinal tree representation
able to use $(2 - \Theta(1)) n$ bits is the so-called \emph{ultra-succinct} 
representation \cite{JSS07}, which uses 
$\sum_a n_a \lg\frac{n}{n_a} + o(n)$ bits, where
$n_a$ is the number of nodes with $a$ children. 
Our empirical simulations suggest that the combination of \cite{fh-sjc11} with
\cite{JSS07} would not use $(2 - \Omega(1)) n$ 
bits on average on random permutations. We generated random
permutations of sizes $10^3$ to $10^7$ and measured the entropy
$\sum_a n_a \lg\frac{n}{n_a}$ on the resulting Cartesian trees. The results,
averaged over 100 to 1,000 iterations, are 
$1.991916$, $1.998986$, $1.999869$, $1.999984$ and $1.999998$,
respectively. The results
appear as a straight line on a log-log plot, which suggests a formula of the 
form $2n-f(n)$ for a very slowly growing function $f(n)$. Indeed, using the
model $2n-O(\log n)$ we obtain the approximation $2n - 0.81 \lg n$ with a
mean squared error below $10^{-9}$.

To understand the observed behaviour, first note that
when the Cartesian tree is converted to an ordinal tree, the arity of each ordinal tree
node $u$ turns out to be, in the Cartesian tree, the length of the path from 
the right child $v$ of $u$ to the leftmost descendant of $u$ (i.e., the node
representing $u+1$ if we identify Cartesian tree nodes with their positions
in $A$). This is called $r_u$ (or $L_v$) in the next section. Next, note that:

\begin{fact} \label{lem:geometric}
The probability that a node $v$ of the Cartesian tree of a random permutation
has a left child is $\frac{1}{2}$.
\end{fact}
\begin{proof}
Consider the values $A[v-1]$ and $A[v]$. If $A[v] < A[v-1]$, then 
$\rmqarg{v-1}{v}=v=\lcaarg{v-1}{v}$, thus $v-1$ descends from $v$ and hence
$v$ has a left child. If $A[v] > A[v-1]$, then
$\rmqarg{v-1}{v}=v-1=\lcaarg{v-1}{v}$, thus $v$ descends from $v-1$ and hence
$v$ is the leftmost node of the right subtree of $v-1$, and therefore $v$ 
cannot have
a left child. Therefore $v$ has a left child iff $A[v] < A[v-1]$, which happens
with probability $\frac{1}{2}$ in a random permutation.
\end{proof}

Thus, if we disregarded the dependencies between nodes in the tree, we
could regard $L_v$ as a geometric variable with parameter $\frac{1}{2}$, and
thus the expected value of $n_a$ would be $\mathbb{E}(n_a)=\frac{n}{2^{a+1}}$. 
Taking the expectation as a fixed value, the space would be 
$\sum_a \mathbb{E}(n_a) \lg\frac{n}{\mathbb{E}(n_a)} =
\sum_{a \ge 0} \frac{n(a+1)}{2^{a+1}} = 2n$.
Although this is only a heuristic argument (as we are ignoring
both the dependencies between tree nodes and the variance of 
the random variables), our empirical results nevertheless 
suggest that this simplified model is 
asymptotically accurate, and thus, 
that no space advantage is obtained by representing 
random Cartesian trees, as opposed to worst-case Cartesian trees, using this 
scheme.

%

\section{Range Top-2 Queries}
\label{sec:secondmin}
In this section we consider a generalization of the RMQ problem.  Again, let 
$A\range{1}{n}$ be an array of elements from a totally ordered set. 
Let $\secondminarg{i}{j}$, for any $1\le i <  j\le n$, denote the position of the
second smallest value in $A\range{i}{j}$.  More formally:
$$\secondminarg{i}{j} = \argsecondmin{k}{i}{j}\;.$$
The \emph{encoding \toptwo{}} problem is to preprocess $A$ into a data structure that,
given $i, j$, returns $\toptwoarg{i}{j} = (\rmqarg{i}{j}, \secondminarg{i}{j})$,
without accessing $A$ at query time. 

The idea is to augment the Cartesian tree of $A$, denoted \cartesian{A}, with some 
information that allows us to answer \secondminarg{i}{j}. 
If $h$ is the position of the minimum element in $A\range{i}{j}$
(i.e., $h=\rmqarg{i}{j}$),  then $h$ divides \range{i}{j} into 
two subranges \range{i}{h-1} and \range{h+1}{j}, and the 
second minimum is the smaller of the elements $A[\rmqarg{i}{h-1}]$ 
and $A[\rmqarg{h+1}{j}]$. Except for the case where one of the 
subranges is empty, the answer to this comparison is not encoded in \cartesian{A}. 
We describe how to succinctly encode the ordering between the elements 
of $A$ that are candidates for \secondminarg{i}{j}. Our data structure 
consists of this encoding together with the
encoding of \cartesian{A} using the representation of Theorem~\ref{thm:bin-rep-inorder}
(along with the operations mentioned in Section~\ref{sec:bin-rep-inorder}).

We define the \emph{left spine} of a node $u$ to be the set of nodes on the 
downward path from $u$ (inclusive) that follows left children until 
this can be done no further. The right spine of a node is defined analogously.
The \emph{left inner spine} of a node $u$ is the right spine of $u$'s left child.
If $u$ does not have a left child then it has an empty left inner spine. The
right inner spine is defined analogously.  We use the notation 
\lspine{v}/\rspine{v}, \lispine{v}/\rispine{v}, $L_v/R_v$ and $l_v/r_v$ to denote
the left/right spines of $v$, the left/right inner spines of $v$, and the
number of nodes in the spines and inner spines of $v$ respectively.  We also
assume that nodes are numbered in inorder and identify node names with their inorder numbers.

As previously mentioned, our data structure encodes the ordering between the candidates for \secondminarg{i}{j}. We first identify locations for these candidates:
\begin{lem}\label{lem:whereis2ndmin}
In \cartesian{A}, for any $i,j \in \range{1}{n}$, $i < j$, $\secondminarg{i}{j}$ is located in
\lispine{v} or \rispine{v}, where $v = \rmqarg{i}{j}$.
\end{lem}
\begin{proof}
Let $v=\rmqarg{i}{j}$.  The second minimum clearly lies in one of two subranges 
\range{i}{v-1} and \range{v+1}{j}, and it must be 
equal to either $\rmqarg{i}{v-1}$ or $\rmqarg{v+1}{j}$.  
W.l.o.g.\ assume that $\range{i}{v-1}$ is non-empty: in this case
the node $v-1$ is the bottom-most node on \lispine{v}. 
Furthermore, since $v = \rmqarg{i}{j}$, $i$ must lie in
the left subtree of $v$.  Since the LCA of the bottom-most 
node on \lispine{v} and any other node in the 
left subtree of $v$ is a node in \lispine{v}, 
\rmqarg{i}{v-1} is in \lispine{v}.
The analogous statement holds for $\rispine{v}$.
\end{proof}
Thus, for any node $v$, it suffices to store the relative order between nodes in
$\lispine{v}$ and $\rispine{v}$ to find \secondminarg{i}{j} for all queries for which
$v$ is the answer to the RMQ query.  As \cartesian{A} determines the ordering among
the nodes of $\lispine{v}$ and also among the nodes of $\rispine{v}$,
we only need to store the information needed to \emph{merge} \lispine{v} and \rispine{v}.
We will do this by storing $m_v = \max(l_v + r_v - 1, 0)$ bits associated with $v$, for
all nodes $v$, as explained later. 
We need to bound the total space required for the `merging' bits, as well
as to space-efficiently realize the association of $v$ with the 
the $m_v$ merging bits associated with it.  For this,
we need the following auxiliary lemmas:
\begin{lem}
\label{lem:sumliri}
Let $T$ be a binary tree with $m$ nodes (of which $m_0$ are leaves) 
and  root $u$. Then,  $\sum_{v \in T} (l_v + r_v) = 2m - L_u - R_u$, and
$\sum_{v \in T} m_v \le m - L_u - R_u + m_0$.
\end{lem}
\begin{proof}
The first part follows from the fact that 
the $R_u$ nodes in $\rspine{u}$ do not appear in $\lispine{v}$ for any
$v \in T$, and all the 
other nodes in $T$ appear exactly once in a left inner spine.
Similarly, the $L_u$ nodes in $\lspine{u}$ do not appear in $\rispine{v}$ for any
$v \in T$, and the other nodes in $T$ appear exactly once in a right inner spine.
Then the second part follows from the fact that $m_v = l_v + r_v - 1$ iff 
$l_v+r_v > 0$, that is, $v$ is not a leaf. If $v$ is a leaf, then $l_v + 
r_v = 0 = m_v$. Thus we must subtract $m - m_0$ from the previous formula, which 
is the number of non-leaf nodes in $T$.
\end{proof}

In the following lemma, we utilize two operations \leftdepth{v} and
\rightdepth{v} which compute the number of nodes that have their left and
right child, respectively, in the path from root to $v$ (recall that \leftdepth{v} computes $c_2(v)$ defined in Section \ref{sec:bin-rep-inorder}).

\begin{lem}
\label{lem:prefixsumliri}
Let $T$ be a binary tree with $m$ nodes and root $\tau$.  Suppose that the
nodes of $T$ are numbered $0,\ldots, m-1$ in inorder.  Then, for any
$0 \le u < m$:

$$\sum_{j < u} (l_j + r_j) = 2 u - L_\tau - l_u + \leftdepth{u} - \rightdepth{u} + 1.$$
\end{lem}
\begin{proof}
The proof is by induction on $m$.  For the base case $m = 1$, $\tau = u = 0$ is the only possibility and 
the formula
evaluates to 0 as expected:
$l_u = \leftdepth{u} = \rightdepth{u} = 0$ and $L_\tau = 1$ (recall that $\tau$ is included
in $\lspine{\tau}$).

Now consider a tree $T$ with root $\tau$ and $m > 1$ nodes.  We consider the three cases
$u = \tau$, $u < \tau$ and $u > \tau$ in that order.  If $u = \tau$ then 
$\leftdepth{\tau} = \rightdepth{\tau} = 0$.  If $\tau$ has no left child, the
situation is the same as when $m=1$. Else, letting $v$ be the left child of
$\tau$, note that
$L_v = L_\tau - 1$ and since $\lispine{\tau} = \rspine{v}$, $l_\tau = R_v$.  As 
the subtree rooted at $v$ has size exactly $\tau$, the formula
can be rewritten as $2 \tau - L_v - R_v$, its correctness  follows from Lemma~\ref{lem:sumliri} without
recourse to the inductive hypothesis.

If $u < \tau$ then by induction on the subtree rooted at the left child $v$
of $\tau$, 
the formula gives $2 u - L_v - l_u + \ensuremath{\text{Ldepth}}'(u) - 
\ensuremath{\text{Rdepth}}'(u) + 1$, where $\ensuremath{\text{Rdepth}}'$
and $\ensuremath{\text{Ldepth}}'$ are measured with respect to $v$.  As
$\ensuremath{\text{Ldepth}}'(u) = \leftdepth{u} - 1$,
$\ensuremath{\text{Rdepth}}'(u) = \rightdepth{u}$ and
$L_v = L_\tau - 1$, this equals $2u - L_\tau - l_u + \leftdepth{u} - \rightdepth{u} + 1$
as required.

Finally we consider the case $u > \tau$.  Letting $v$ and $w$ be the left and
right children of $\tau$, and 
$u' = u - \tau - 1$, we note that $u'$ is the inorder number of $u$ in the subtree rooted at $w$.
Applying the induction hypothesis to the subtree rooted at $w$, we get that:

$$
\sum_{\tau < j < u} (l_j + r_j) = 2 u' - L_w - l_u + \ensuremath{\text{Ldepth}}'(u) - 
\ensuremath{\text{Rdepth}}'(u) + 1,
$$
where $\ensuremath{\text{Rdepth}}'$
and $\ensuremath{\text{Ldepth}}'$ are measured with respect to $w$.  Simplifying:
\begin{eqnarray*}
\makebox[0.6cm][l]{$\sum_{j < u} (l_j + r_j)~~=~~ \sum_{j < \tau} (l_j + r_j)
+ l_\tau + r_\tau + \sum_{\tau < j < u} (l_j + r_j)$}  & & \\
                         & = & 2\tau - L_v - R_v + l_\tau + r_\tau + 2 u' - L_w - l_u + \ensuremath{\text{Ldepth}}'(u) - \ensuremath{\text{Rdepth}}'(u) + 1\\
                         & = & 2\tau - L_v - R_v + l_\tau + r_\tau + 2 u' - L_w - l_u + 
                               \leftdepth{u} - \rightdepth{u} + 2\\
                         & = & 2\tau - L_v + 2 u' - l_u + 
                               \leftdepth{u} - \rightdepth{u} + 2\\
                         & = & 2\tau - (L_\tau - 1) + 2(u - \tau - 1) - l_u +  
	                           \leftdepth{u} - \rightdepth{u} + 2\\
                         & = & 2u - L_\tau - l_u + \leftdepth{u} - \rightdepth{u} + 1
\end{eqnarray*}
Here we have made use (in order) of Lemma~\ref{lem:sumliri} and 
the facts $\ensuremath{\text{Ldepth}}'(u) = \leftdepth{u}$
and $\ensuremath{\text{Rdepth}}'(u) = \rightdepth{u} - 1$; $L_w = r_\tau$ and
$R_v = l_\tau$; and 
finally $L_v = L_\tau - 1$.  \end{proof}

\begin{cor}
\label{cor:prefixsumliri}
In the same scenario of Lemma~\ref{lem:prefixsumliri}, we have
$$\sum_{j < u} m_j ~=~ 2 u - L_\tau - l_u + \leftdepth{u} - \rightdepth{u} + 1
		- \leftleaves{u},$$
where $\leftleaves{u}$ is the number of leaves to the left of node $u$.
\end{cor}
\begin{proof}
Trivially follows from Lemma~\ref{lem:prefixsumliri} and the same considerations as in the
proof of Lemma~\ref{lem:sumliri}.
\end{proof}

\paragraph{The Data Structure.}
For each node $u$ in \cartesian{A}, we create a bit sequence $M_u$ of length $m_u$ to
encode the merge order of \lispine{u} and \rispine{u}. $M_u$ is obtained by
taking the sequence of all the elements of $\lispine{u} \cup \rispine{u}$ sorted in decreasing
order, and replacing each element of this sorted sequence by 0 if the element
comes from \lispine{u} and by 1 if the element comes from \rispine{u} (the
last bit is omitted, as it is unnecessary). We concatenate the bit sequences $M_u$ for all
$u \in \cartesian{A}$ considered in inorder and call the concatenated sequence $M$.
  
The data structure comprises $M$, augmented with $\rank$ and $\select$ operations and a 
data structure for \cartesian{A}.  If we use Theorem~\ref{thm:bin-rep-inorder}, then
\cartesian{A} is represented in $2n + o(n)$ bits, and the (augmented) $M$ takes at most $1.5n + o(n)$ bits by Lemmas~\ref{lem:sumliri} and~\ref{lem:rank-select}, 
since there are at most $(n + 1)/2$ leaves in an $n$-node binary tree. This gives a
representation whose space is $3.5n + o(n)$ bits.
A further improvement can be obtained by using Theorem~\ref{thm:general} as follows.
For some real parameter $0 < x < 1$, consider the concave function:
$$H(x) ~~=~~
	2 x \lg \frac{1}{x} + 2 \frac{(1-2x)}{2} \lg \frac{2}{1-2x} +x + 1.$$
Differentiating and simplifying, we get the maximum of $H(x)$ as the
solution to the equation $2 (\lg(1-2x) - \lg x) = 1$, from which we
get that $H(x)$ is maximized at $x = 1 - \sqrt{2}/2 \approx 0.293$, and attains a 
maximum value of $\gamma = 2 + \lg (1+\sqrt{2}) < 3.272$.   

Now let $n_0, n_L (n_R)$ and $n_2$ be the numbers of leaves, 
nodes with only a left (right) child and nodes with
both children in \cartesian{A}. Letting $x = n_0/n$, we apply
Theorem~\ref{thm:general} to represent \cartesian{A}, using the parameters 
$\alpha_0 = \alpha_2$ to be equal to $x$, but capped to a minimum of $0.05$
and a maximum of $0.45$, i.e. $\alpha_0 = \alpha_2 = \max\{\min\{0.45, x\}, 0.05\}$,
and $\alpha_L = \alpha_R = (1 - 2 \alpha_0)/2$.  Observe that the capping means
that $\alpha_L$ and $\alpha_R$ lie in the range $[0.05, 0.45]$ as well, thus
satisfying the condition in Theorem~\ref{thm:general} requiring the $\alpha_i$'s to be constant. 
Then the space used by the representation is
$\left ( \sum_{i \in \{0,L,R,2\}} n_i \lg (1/\alpha_i) \right) + n + n_0 + o(n)$ bits.
Provided capping is not applied, 
and since $n_0 = n_2 + 1$ and $\alpha_L = \alpha_R$, 
this is easily seen to be $n H(x) + o(n)$ bits,
and is therefore bounded by $\gamma n + o(n)$ bits.  If $x > 0.45$,
then the representation takes 
$2 n_0 \lg (1/0.45) + (n - 2 n_0) \lg (1/0.05) + n + n_0 + o(n)$ bits.
Since $2\lg(1/0.45)-2\lg(1/0.05)+1 < 0$, this is maximized with the least 
possible $n_0 = 0.45n$, where the space is precisely 
$nH(0.45)+o(n) < \gamma n + o(n)$.
Similarly, for $x < 0.05$ the space is less than 
$nH(0.05)+o(n) < \gamma n + o(n)$ bits.


\medskip

We now explain how this data structure can answer \toptwo{} in constant time.  
We utilize the data structure of Theorem~\ref{thm:general} constructed on
$\cartesian{A}$ in order to find $u=\lcaarg{{i}}{{j}} =\rmqarg{i}{j}$.  Subsequently:

\begin{enumerate}
\item We determine the start of $M_u$ within $M$ by calculating  
$\sum_{j < u} m_j$.
\item We locate the appropriate nodes from $\lispine{u}$ and $\rispine{u}$ and the 
corresponding bits within $M_u$ and make the required comparison.
\end{enumerate}
We now explain each of these steps.  For step (1), we use Corollary~\ref{cor:prefixsumliri}.  When
evaluating the formula,  we use the $O(1)$-time 
support for $\leftdepth{u}$ and $\rightdepth{u}$ given by the data structure of Section \ref{sec:bin-rep-inorder}; there we explain $\leftdepth{u}$ indeed computes $c_2(u)$ and we describe how to compute $c_2(u)$ in constant time (computing \rightdepth{u} can be done analogously). This leaves only the computation
of $l_u$ and $\leftleaves{u}$. The former is done as follows.  We check if $u$ has a left child: if not,
then $l_u = 0$. Otherwise, if $v$ is $u$'s left child, then $v$ and $u-1$ are respectively the
topmost and lowest nodes in $\lispine{u}$.  We can then obtain $l_u$ in $O(1)$ time
as $\deptharg{v} - \deptharg{u}$ in $O(1)$ time by 
Theorem~\ref{thm:general}.
On the other hand, $\leftleaves{u}$ can be computed as 
$\text{leaf-rank}(v'+\text{subtree-size}(v')-1)$, where 
$v'=\text{node-select}_\mathrm{inorder}(v)$ and $v$ is the left child of $u$.
If $v$ does not exist then $\leftleaves{u} = \text{leaf-rank}(u')$, where
$u'=\text{node-select}_\mathrm{inorder}(u)$. All those operations take $O(1)$
time by Theorem~\ref{thm:general}.

For step (2) we use Lemma~\ref{lem:whereis2ndmin} to locate the two candidates
from $A\range{i}{u-1}$ and $A\range{u+1}{j}$ (assuming that $i < u < j$, if not, 
the problem is easier) in $O(1)$ time as $v = \lcaarg{i}{u-1}$ and $w = \lcaarg{u+1}{j}$.
Next we obtain the rank $\rho_v$ of $v$ in $\lispine{u}$ in $O(1)$ time as 
$\deptharg{u-1} - \deptharg{v}$. The rank $\rho_w$ of $w$ in $\rispine{u}$ is obtained similarly.
Now, letting $\Delta = \sum_{j < u} (l_j + r_j)$, we compare $\selectarg{M}{0}{\rankarg{M}{0}{\Delta} + \rho_v}$
and $\selectarg{M}{1}{\rankarg{M}{1}{\Delta} + \rho_w}$ in $O(1)$ time to determine which of $v$ and $w$ is smaller
and return that as the answer to $\secondminarg{i}{j}$.\footnote{If we select
the last (non-represented) bit of $M_u$, the result will be out of the $M_u$
area of $M$, but nevertheless the result of the comparison will be correct.}  
We have thus shown:
\begin{theorem}
\label{thm:secondmin}
Given an array of $n$ elements from a totally ordered set, there exists a data
structure of size at most $\gamma n+o(n)$ bits that supports \toptwo{s} in $O(1)$ time,
where $\gamma = 2 + \lg (1+\sqrt{2}) < 3.272$. \end{theorem}

Note that $\gamma n$ is a worst-case bound. The size of the encoding can be 
less for other values of $n_0$. In particular, since $H(x)$ is convex and
the average value of $n_0$ on random permutations is $n/3$ \cite[Theorem
1]{gikrs-isaac11}, we have by Jensen's inequality that the expected size of 
the encoding is below $H(1/3) = \lg(3)+\frac{5}{3} < 3.252$.

\section{Effective Entropy of \toptwo{} and \secondmin{}}

In this section we lower bound the effective
entropy of \toptwo{}, that is, the number of equivalence classes $\mathcal{C}$
of arrays distinguishable by \toptwo{s}. For this sake, we define 
{\em extended} Cartesian trees, in which each node $v$ indicates a 
merging order between its left and right internal spines, using 
a number in a universe of size ${l_v + r_v \choose r_v}$. We prove that any
distict extended Cartesian tree can arise for some input array, and that any
two distinct extended Cartesian trees give a different answer for at least
some \toptwo{}. Then we aim to count the number of distinct extended Cartesian 
trees. 

While unable to count the exact number of extended Cartesian trees, we provide 
a lower bound by unrolling their recurrence a finite number of times
(precisely, up to 7 levels). This effectively limits the lengths of internal
spines we analyze, and gives us a number of configurations of the 
form $\frac{1}{0.160646^n}\,\theta(n)$ for a polynomial $\theta(n)$, 
from where we obtain a lower bound of 
$2.638n - O(\log n)$ bits on the effective entropy of \toptwo{}.

We note that our bound on \toptwo{s} also applies to the weaker \secondmin{}
operation, since any encoding answering \secondmin{s} has enough information
to answer \toptwo{s}. Indeed, it is easy to see that RMQ$(i,j)$ is the only 
position that is not the answer of any query \secondmin{}$(i',j')$ for any 
$i \le i' < j' \le j$. Then, with RMQ and \secondmin{}, we have \toptwo{}.
Therefore we can give our result in terms of the weaker \secondmin{}. 

\begin{theorem} \label{thm:lbsecondmin}
The effective entropy of \secondmin{} (and \toptwo{})
over an array $A[1,n]$ is at least $2.638n - O(\log n)$.
\end{theorem}

\subsection{Modeling the Effective Entropy of \secondmin{}}

Recall that to show that the effective entropy of RMQ is $2n - O(\log n)$ bits, we
argue that $(i)$ any two Cartesian trees will give a different 
answer to at least one $\rmqarg{i}{j}$; $(ii)$ any binary tree is the Cartesian
tree of some permutation $A[1,n]$; $(iii)$ the number of binary trees of $n$ 
nodes is $\frac{1}{n+1}{2n \choose n}$, thus in the worst case one needs at 
least $\lg  \left(\frac{1}{n+1}{2n \choose n}\right) = 2n - O(\log n)$ bits to 
distinguish among them. 

A similar reasoning can be used to establish a lower bound on the
effective entropy of \toptwo{}. We consider an extended Cartesian tree $T$ of size
$n$, where for any node $v$ having both left and right children we store a 
number $M(v)$ in the range $[1..{l_v+r_v \choose r_v}]$. The number $M(v)$ 
identifies one particular merging order between the nodes in lispine$(v)$ and 
rispine$(v)$, and ${l_v+r_v \choose r_v}$ is the exact number of different
merging orders we can have.

Now we follow the same steps as before. For $(i)$, let $T$ and $T'$ be Cartesian
trees extended with the corresponding numbers $M(v)$ for $v\in T$ and $M'(v')$
for $v' \in T'$. We already know that if the topologies of $T$ and $T'$
differ, then there exists an $\rmqarg{i}{j}$ that gives different
results on $T$ and $T'$. Assume now that the topologies are equal, but there
exists some node $v$ where $M(v)$ differs from $M'(v)$. Then there exists an
$\toptwoarg{i}{j}$ where the extended trees give a different result. 
W.l.o.g., let $i$ and $j$ be the first positions of lispine$(v)$ and
rispine$(v)$, respectively, where $v_l = \textrm{lispine}(v)[i]$ goes before 
$v_r = \textrm{rispine}(v)[j]$ according to $M(v)$, but after according to
$M'(v)$. Then $T$ answers $\secondminarg{v_1}{v_2}=v_1$ and $T'$ answers 
$\secondminarg{v_1}{v_2}=v_2$ (we interpret $v_1$ and $v_2$ as inorder
numbers here).

As for $(ii)$, let $T$ be an extended Cartesian tree, where $u$ is the (inorder
number of the) root of $T$. Then we build a permutation $A[1,n]$ whose
extended tree is $T$ as follows. First, we set the minimum at $A[u]=0$. Now, 
we recursively build the ranges $A[1,u-1]$ (a permutation in with values in 
$[0..u-1]$) and $A[u+1,n]$ (a permutation with values in $[0..n-u-1]$) for the 
left and right child of $T$, respectively. Assume, inductively, that the
permutations already satisfy the ordering given by the $M(v)$ numbers for all
the nodes $v$ within the left and right children of $u$. Now we are free to
map the values of $A \setminus A[u]$ to the interval $[1,n-1]$ in any way that
maintains the relative ordering within $A[1,u-1]$ and $A[u+1,n]$. We do so in
such a way that the elements of lispine$(u)$ and rispine$(u)$ compare according
to $M(u)$. This is always possible: We sort $A[1,u-1]$ and $A[u+1,n]$
from smallest to largest values, let $A[a_i]$ be the $i$th smallest cell of 
$A[1,u-1]$ and $A[b_j]$ the $i$th smallest cell of $A[u+1,n]$. Also, we set 
cursors at lispine$(u)[l]$ and rispine$(u)[r]$, initially $l=r=1$, and set
$c=i=j=1$. At each step, if $M(u)$ indicates that lispine$(u)[l]$ comes before
rispine$(u)[r]$, we reassign $A[a_i]=c$ and increase $i$ and $c$, until
(and including) the reassignment of $a_i = \textrm{lispine}(u)[l]$, then we 
increase $l$; otherwise we reassign $A[b_j]=c$ and increase $j$ and $c$, until 
(and including) the reassignment of $b_j = \textrm{rispine}(u)[r]$, then we
increase $r$. We repeat the process until reassigning all the values in $A
\setminus A[u]$.

For $(iii)$, next we will lower bound the total number of extended Cartesian
trees.

\subsection{Lower Bound on Effective Entropy}

As explained, we
have been unable to come up with a general counting for the lower bound, yet
we give a method that can be extended with more and more effort to reach
higher and higher lower limits.
The idea is to distinguish the first steps in the generation of the
Cartesian tree out of the root node, and charge the minimum value of 
${l_v+r_v \choose r_v}$ we can ensure in each case. Let 
\[ T(x) ~~=~~ \sum_{n > 0} t(n)x^n
\]
where $t(n)$ is the number of extended Cartesian trees with $n$ nodes, counted
using some simple lower-bounding technique. For example, if we consider the
simplest model for $T(x)$, we have that a (nonempty) tree is a root $v$ either
with no children, with a left child rooting a tree, with a right child 
rooting a tree, or with left and right children rooting trees, this
time multiplied by 2 to account for ${l_v+r_v \choose r_v} \ge
{2 \choose 1}$ (see the levels 0 and 1 in Figure~\ref{fig:lb}). 
Then $T(x)$ satisfies
\[ T(x) ~~=~~ x + xT(x) + xT(x) + 2xT(x)^2 ~~=~~ x + 2xT(x) + 2xT(x)^2,
\]
which solves to
\[ T(x) ~~=~~ \frac{1-2x - \sqrt{1-4x-4x^2}}{4x},
\]
which has two singularities at
$x=\frac{-1\pm\sqrt{2}}{2}$. The one closest to the origin is
$x=\frac{\sqrt{2}-1}{2}$. Thus it follows that $t(n)$ is of the form
$\left(\frac{2}{\sqrt{2}-1}\right)^n \theta(n)$ for some polynomial
$\theta(n)$ \cite{SF95}, and thus we need at least
$\lg\left(\left(\frac{2}{\sqrt{2}-1}\right)^n \theta(n)\right) =
\lg\left(\frac{2}{\sqrt{2}-1}\right)n - O(\log n) \ge 2.271n-O(\log n)$ 
bits to represent all the possible extended Cartesian trees.

\begin{figure}[t]
\begin{center}
\includegraphics[width=\textwidth]{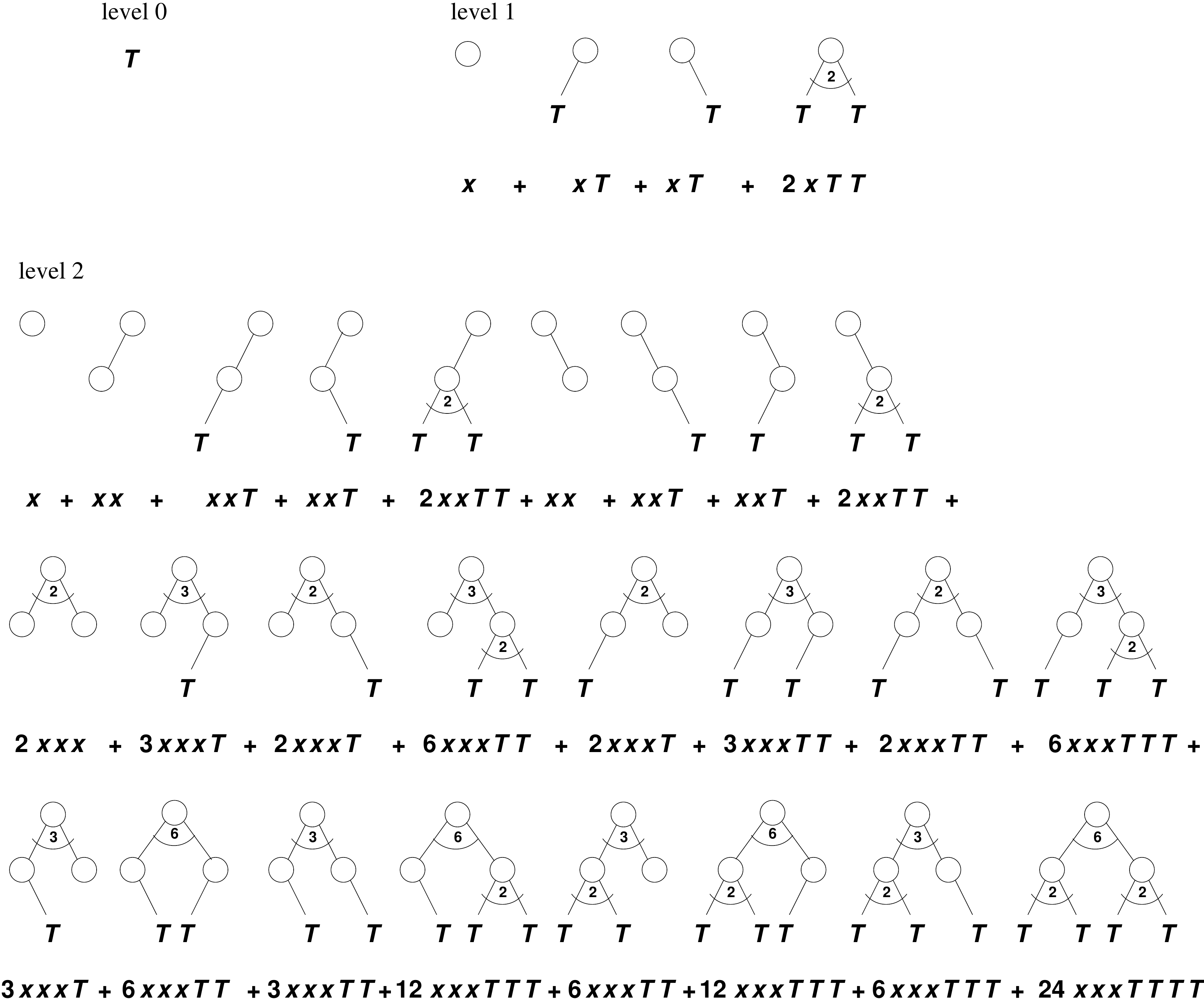}
\end{center}
\caption{Our scheme to enumerate extended Cartesian trees $T$ with increasing
detail, where the $x$ stands for a node and $T$ for any subtree. We indicate
the numbers ${l_v+r_v \choose r_v}$ below nodes having left and right internal
spines. Level 0 
corresponds just to $T(x)$. In level 1 we have four possibilities, which lead 
to the equation $T(x) = x + 2xT(x) + 2xT(x)^2$. For level 2, each of the $T$s 
in level 1 is expanded in all the four possible ways, leading to 25 
possibilities and to the equation
$T(x) = x + 2x^2 + 4x^2 T(x) + 4x^2 T(x)^2 + 2x^3 + 10 x^3 T(x) + 26 x^3 T(x)^2
+ 36 x^3 T(x)^3 + 24 x^3 T(x)^4$. }
\label{fig:lb}
\end{figure}

This result can be improved by unrolling the recurrence of $T$ further, that
is, replacing each $T$ by its four possible alternatives in the basic
definition. Then the lower bound improves because 
some left and right internal spines can be seen to have length two or more.
The results do not admit easy algebraic solutions anymore, but we can 
numerically analyze the resulting functions with Maple and establish a safe
numeric threshold from where higher lower bounds can be derived. For example
by doing a first level of replacement in the simple recurrence, we obtain a
recurrence with 25 cases, which yields
{\footnotesize
\[ T(x) ~=~
x + 2x^2 + 4x^2 T(x) + 4x^2 T(x)^2 + 2x^3 + 10 x^3 T(x) + 26 x^3 T(x)^2
+ 36 x^3 T(x)^3 + 24 x^3 T(x)^4;
\]
}
%
%
\noindent
(see level 2 in Figure~\ref{fig:lb})
which Maple is able to solve algebraically, although the formula is too long 
to display it here. While Maple could not algebraically find the singularities 
of $T(x)$, we analyzed the result numerically and found a singularity at 
$x=0.190879...$ Therefore, we conclude that $t(n) \ge 
\frac{1}{0.190880^n}\theta(n)$, and thus that a lower bound is 
$n\,\lg\frac{1}{0.190880} - O(\log n) \ge 2.389n-O(\log n)$.

To find the singularity we used the result \cite[Thm.~VII.3]{SF09}
that, under certain conditions that are met in our case, the singularities of 
an equation of the form $T(x) = G(x,T(x))$ can be found by numerically solving 
the system formed by the equation $T = G(x,T)$ and its derivative, 
$1 = \frac{\partial G(x,T)}{\partial T}$. If the positive solution is found at
$(x=r,T=\gamma)$, then there is a singularity at $x=r$. If, further, $T(x)$
is aperiodic (as in our case), then $r$ is the unique dominant singularity
and $t(n) = \frac{1}{r^n}\,\theta(n)$ for some polynomial $\theta(n)$.%

To carry the idea further, we wrote a program that generates all the
combinations of any desired level, and builds a recurrence to feed Maple with. 
We use the program to generate the recurrences of level 3 onwards. 
Table~\ref{tab:lbs} shows the results obtained up to level 7, which is the
one yielding the lower bound $2.638n-O(\log n)$ of
Theorem~\ref{thm:lbsecondmin}. 
This was not without challenges; we describe the details in the Appendix.

\begin{table}[t]
\begin{center}
\begin{tabular}{rrrrrr}
Level & \# of cases & \# of terms & degree & singularity & lower bound~~~~~~~ \\
\hline
1 &    4                      &    3 &   2 & 0.207107 & $2.271n-O(\log n)$\\
2 &   25                      &    9 &   4 & 0.190879 & $2.389n-O(\log n)$\\
3 &  675                      &   63 &   8 & 0.179836 & $2.474n-O(\log n)$\\
4 & $\sim 4.6 \times 10^5$    &  119 &  16 & 0.172288 & $2.537n-O(\log n)$\\
5 & $\sim 2.1 \times 10^{11}$ &  479 &  32 & 0.167053 & $2.581n-O(\log n)$\\
6 & $\sim 4.4 \times 10^{22}$ & 1951 &  64 & 0.163343 & $2.621n-O(\log n)$\\
7 & $\sim 1.9 \times 10^{45}$ & 7935 & 128 & 0.160646 & $2.638n-O(\log n)$\\
\end{tabular}
\end{center}
\caption{Our results for increasing number of levels. The second column gives
the number of cases generated, the third the number of terms in the resulting
polynomial, the fourth the degree of the polynomial in $x$ and $T$, the fifth
the $x$ value of the singularity found, and the last column gives the implied
lower bound.}
\label{tab:lbs}
\end{table}

\section{Conclusions}
\label{sec:concl}

We obtained a succinct binary tree representation that extends the
representation of Farzan and Munro~\cite{fm-algo12} by supporting navigation
based on the inorder numbering of the nodes, and a few additional operations.
Using this representation, we describe how to encode an array in optimal space
in a more natural way than the existing structures, to support RMQs in constant time. 
In addition, this representation reaches $1.919n+o(n)$ bits on random
permutations, thus breaking the worst-case lower bound of $2n-O(\log n)$ bits. 
This is not known to hold on any alternative representation. It is an open question to 
find a data structure that answers RMQs in $O(1)$ time using $2n + o(n)$ bits 
in the worst case, while also achieving the expected
effective entropy bound of about $1.736n$ bits for random arrays $A$.

Then, we obtain another structure that encodes an array of $n$ elements from a
total order using $3.272n+o(n)$ bits to support \toptwo{s} in
$O(1)$ time. This uses almost half of the $6n+o(n)$ bits used for this
problem in the literature \cite{GINRS13}. Our structure can possibly be plugged
in their solution, thus reducing their space.

While the effective entropy of RMQs is known 
to be precisely $2n- O(\log n)$ bits, the effective entropy for 
range top-$k$ queries is only known asymptotically: it is at 
least $n \lg k - O(n)$ bits, and at most $O(n \log k)$ bits \cite{GINRS13}.
We have shown that, for $k=2$, the effective entropy is 
at least $2.638n-O(\log n)$ bits.  
Determining the precise effective
entropy for $k\ge 2$ is an open question.

\section*{Acknowledgements}
Many thanks to Jorge Olivos and Patricio Poblete for discussions (lectures) on extracting asymptotics from generating functions.

\bibliographystyle{plain}
\bibliography{paper}

\begin{thebibliography}{10}

\bibitem{Cla96}
D.~Clark.
\newblock {\em Compact Pat Trees}.
\newblock PhD thesis, University of Waterloo, Canada, 1996.

\bibitem{CLR}
Thomas~H. Cormen, Charles~E. Leiserson, Ronald~L. Rivest, and Clifford Stein.
\newblock {\em Introduction to Algorithms}.
\newblock The MIT Press, 2 edition, 2001.

\bibitem{fm-swat08}
Arash Farzan and J.~Ian Munro.
\newblock A uniform approach towards succinct representation of trees.
\newblock In {\em Proc. 11th Scandinavian Workshop on Algorithm Theory}, volume
  5124 of {\em LNCS}, pages 173--184. Springer-Verlag, 2008.

\bibitem{fm-algo12}
Arash Farzan and J.~Ian Munro.
\newblock A uniform paradigm to succinctly encode various families of trees.
\newblock {\em Algorithmica, to appear}, 2012.

\bibitem{fh-sjc11}
Johannes Fischer and Volker Heun.
\newblock Space-efficient preprocessing schemes for range minimum queries on
  static arrays.
\newblock {\em SIAM Journal on Computing}, 40(2):465--492, 2011.

\bibitem{SF09}
P.~Flajolet and R.~Sedgewick.
\newblock {\em Analytic Combinatorics}.
\newblock Cambridge University Press, 2009.

\bibitem{gbt-stoc84}
Harold~N. Gabow, Jon~Louis Bentley, and Robert~E. Tarjan.
\newblock Scaling and related techniques for geometry problems.
\newblock In {\em Proc. 16th annual ACM Symposium on Theory of Computing},
  pages 135--143. ACM Press, 1984.

\bibitem{grr-atalg06}
Richard~F. Geary, Rajeev Raman, and Venkatesh Raman.
\newblock Succinct ordinal trees with level-ancestor queries.
\newblock {\em ACM Transactions on Algorithms}, 2(4):510--534, 2006.

\bibitem{GIKRRS12}
M.~Golin, J.~Iacono, D.~Krizanc, R.~Raman, S.~Srinivasa Rao, and S.~Shende.
\newblock Encoding 2{D} range maximum queries.
\newblock {\em CoRR}, 1109.2885v2, 2012.

\bibitem{gikrs-isaac11}
M.~J. Golin, John Iacono, Danny Krizanc, Rajeev Raman, and S.~Srinivasa Rao.
\newblock Encoding {2D} range maximum queries.
\newblock In {\em Proc. 22nd International Symposium on Algorithms and
  Computation}, volume 7074 of {\em LNCS}, pages 180--189. Springer-Verlag,
  2011.

\bibitem{GINRS13}
R.~Grossi, J.~Iacono, G.~Navarro, R.~Raman, and S.~Srinivasa Rao.
\newblock Encodings for range selection and top-$k$ queries.
\newblock In {\em Proc. 21st Annual European Symposium on Algorithms (ESA)},
  LNCS 8125, pages 553--564, 2013.

\bibitem{hms-icalp07}
Meng He, J.~Ian Munro, and S.~Srinivasa Rao.
\newblock Succinct ordinal trees based on tree covering.
\newblock In {\em Proc. 34th International Colloquium on Automata, Languages
  and Programming}, pages 509--520. Springer-Verlag, 2007.

\bibitem{jacobson89}
Guy Jacobson.
\newblock {\em Succinct Static Data Structures}.
\newblock PhD thesis, Carnegie Mellon University, Pittsburgh, PA, USA, 1989.

\bibitem{JSS07}
J.~Jansson, K.~Sadakane, and W.-K. Sung.
\newblock Ultra-succinct representation of ordered trees.
\newblock In {\em Proc. 18th Annual ACM-SIAM Symposium on Discrete Algorithms
  (SODA)}, pages 575--584, 2007.

\bibitem{Kieffer2009}
John~C. Kieffer, En-Hui Yang, and Wojciech Szpankowski.
\newblock Structural complexity of random binary trees.
\newblock In {\em Proc. IEEE International Symposium on Information Theory
  (ISIT)}, pages 635--639, 2009.

\bibitem{Mun96}
I.~Munro.
\newblock Tables.
\newblock In {\em Proc. 16th Conference on Foundations of Software Technology
  and Theoretical Computer Science (FSTTCS)}, LNCS 1180, pages 37--42, 1996.

\bibitem{mrs-soda01}
J.~Ian Munro, Venkatesh Raman, and Adam~J. Storm.
\newblock Representing dynamic binary trees succinctly.
\newblock In {\em Proc. 12th Annual ACM-SIAM Symposium on Discrete Algorithms},
  pages 529--536. SIAM, 2001.

\bibitem{rrr-talg07}
Rajeev Raman, Venkatesh Raman, and Srinivasa~Rao Satti.
\newblock Succinct indexable dictionaries with applications to encoding {\it
  k}-ary trees, prefix sums and multisets.
\newblock {\em ACM Transactions on Algorithms}, 3(4):Article 43, 2007.

\bibitem{ianfest-survey}
Rajeev Raman and Srinivasa~Rao Satti.
\newblock Succinct representations of ordinal trees.
\newblock In {\em Proc. Conference on Space Efficient Data Structures, Streams
  and Algorithms}, volume 8066 of {\em LNCS}, pages 319--332. Springer-Verlag,
  2013.

\bibitem{SF95}
R.~Sedgewick and P.~Flajolet.
\newblock {\em An Introduction to the Analysis of Algorithms}.
\newblock Addison-Wesley, 1995.

\bibitem{Vuillemin1980}
Jean Vuillemin.
\newblock A unifying look at data structures.
\newblock {\em Communications of the ACM}, 23(4):229--239, 1980.

\end{thebibliography}
\appendix

\section{Unrolling the Lower Bound Recurrence}


The main issue to unroll further levels of the recurrence is that it grows
very fast. The largest tree at level $\ell$ has $2^\ell$ leaves labeled $T$. 
Each such leaf is expanded in 4 possible ways to obtain the trees of the next 
level. Let $A(\ell)$ be the number of trees generated at level $\ell$. If all 
the $A(\ell-1)$ trees had $2^{\ell-1}$ leaves labeled $T$, then we would have 
$A(\ell) = A(\ell-1)\cdot 4^{2^{\ell-1}} \le 2^{2^{\ell+1}}$. 
If we consider just one tree of level $\ell$ with $2^\ell$ leaves labeled $T$, 
we have $A(\ell) = 4^{2^{\ell-1}} = 2^{2^\ell}$. Thus the number of trees to 
generate is $2^{2^\ell} \le A(\ell) \le 2^{2^{\ell+1}}$. For levels 3 and 4
we could just generate and add up all the trees, but from level 5 onwards
we switched to a dynamic programming based counting that performs
$O(\ell^4 \cdot 16^\ell)$ operations, 
which completed level 5 in 40 seconds instead of 4 days 
of the basic method. It also completed level 6 in 20 minutes 
and level 7 in 10 hours. 
We had to use unbounded integers,\footnote{With the 
{\em GNU Multiple Precision Arithmetic Library (GMP)}, at 
{\tt http://gmplib.org}.} since 64-bit numbers overflow already in level 5
and their width doubles every new level. Apart from this, the degree of the
generated polynomials doubles at every new level and the number of terms grows
by a factor of up to 4, putting more pressure on Maple. In level 3, with 
polynomials of degree 8, Maple is already unable to algebraically solve the 
equations related to $G(x,T)$, but they can still be solved numerically.
Since level 5, Maple was unable to solve the system of two equations, and
we had to find the singularity by plotting the implicit function and 
inspecting the axis $x \in [-1,1]$.\footnote{Note that, in principle,
there is a (remote) chance of us missing the dominant singularity by visual 
inspection, finding one farther from the origin instead. Even in this case, 
each singularity implies a corresponding exponential term in the growth of the 
function, and thus we would find a valid lower bound.}
Since level 6, Maple could not even plot the implicit function, and we had
to manually find the solution of the two equations on $G(x,T)$.
At this point even loading the equation into Maple is troublesome; for example
in level 7 we had to split the polynomial into 45 chunks to avoid Maple to
crash.

For level 8, our generation program would take nearly two weeks.
It is likely that Maple would also give problems with the large number of terms
in the polynomial (expected to be near 32000). For level 9 (expected to take
more than one year), we cannot compile as we reach 
an internal limit of the library to handle large integers: The space
usage of the dynamic programming tables grows as $O(\ell^2 \cdot 4^\ell)$ and
for level 9 it surpasses $2^{30}$ large integers. Thus we are very close to 
reaching various limits of practical applicability of this technique. A 
radically different model is necessary to account for every possible internal 
spine length and thus obtain the exact lower bound.

\end{document}